\documentclass[a4paper,12pt]{article}
\usepackage{amssymb,amsmath,amsthm}
\usepackage{texdraw}
\usepackage{a4wide}

\newtheorem{theorem}{Theorem}[section]

\newtheorem{proposition}[theorem]{Proposition}

\newtheorem{corollary}[theorem]{Corollary}

\theoremstyle{definition}
\newtheorem{definition}[theorem]{Definition}

\newtheorem*{unnumrem}{Remark}

\begin{document}

%\begin{titlepage}

\title{Entwining Yang-Baxter maps and integrable lattices}

\author{Theodoros E. Kouloukas and Vassilios G. Papageorgiou}

\maketitle
\begin{center}
\emph{Department of Mathematics, University of Patras \\ GR-265 00 Patras, \
Greece}
\end{center}

\bigskip

\begin{abstract}
Yang--Baxter (YB) map systems (or set-theoretic analoga of entwining YB structures) are presented.
They admit zero curvature representations with spectral parameter 
depended Lax triples $L_1, \ L_2, \ L_3$ derived from symplectic leaves of $2 \times 2$ binomial matrices 
equipped with the Sklyanin bracket. A unique factorization condition of the Lax triple  
implies a 3-dimensional compatibility property of these maps. In case  $L_1=L_2=L_3$ this property yields the set-theoretic 
quantum Yang-Baxter equation, i.e. the YB map property. By considering periodic `staircase' initial value problems on quadrilateral 
lattices, these maps give rise to multidimensional integrable mappings which preserve the spectrum of the corresponding 
monodromy matrix.   
\end{abstract}
%\end{titlepage}
%\newpage

\section{Introduction}
The connection between set theoretical solutions of the quantum Yang-Baxter equations
and integrable mappings has been recently investigated by many authors.
First Veselov in \cite{ves2,ves3} proved that for such
a solution, admitting a Lax matrix, there is a family of
commuting transfer maps which preserve the spectrum of the
corresponding monodromy matrix. He also proposed the shorter term `Yang Baxter map'
for a set-theoretic solution of the quantum Yang-Baxter equation. In the present article we
present system of maps that admit a Lax triple of matrices 
$L_1,\ L_2$, $L_3$. They satisfy 3-dimensional compatibility conditions that constitute 
set-theoretic analoga of entwining quantum Yang-Baxter equations. The latter term was 
introduced in \cite{Nich} for systems of quantum YB equations. Such entwining structures 
for the quantum Yang-Baxter equation appeared already in \cite{Frei-Mail, cnp, vlad} and their study continued in 
\cite{hlav, hlav2}. 

The basic definitions and notations about Lax pairs and integrable mappings are introduced in section 
\ref{secLax}.The transfer map and the monodromy matrix are
defined for a periodic `staircase' initial value problem on a lattice for each pair $L_i, ~ L_j$ and 
integrals of motion are obtained from the spectrum of the monodromy matrix, which are in involution with respect to the Sklyanin bracket.
Section \ref{sec3-d} deals with the the 3-d compatibility condition (entwining Yang-Baxter equation) 
of maps that admit Lax pairs giving rise to Lax triples. 
A construction of Lax triples and their corresponding set theoretical solutions of the entwining YB equation is presented in section \ref{secbin}, while in the next section the whole theory is applied to a concrete example. 
We end  in section \ref{persp} by giving some conclusions and perspectives for future work on this subject.
%%%%%%%%%%%%%%%%%%%%%%%%%%%%%%%%%%%%%%%%%%%%%%%%%%%%%%%%%%%%%%%%%%%%%%%%%%%%%%%%%%%%%%%%%%%%%
%%%%%%%%%%%%%%%%%%%%%%%%%%%%%%%%%%%%%%%%%%%%%%%%%%%%%%%%%%%%%%%%%%%%%%%%%%%%%%%%%%%%%%%%%%%%%
%%%%%%%%%%%%%%%%%%%%%%%%%%%%%%%%%%%%%%%%%%%%%%%%%%%%%%%%%%%%%%%%%%%%%%%%%%%%%%%%%%%%%%%%%%%%%

\section{Lax pairs and integrable mappings} \label{secLax}
Let $S$ be the map
\begin{equation} \label{mapphi}
S:((x,\alpha),(y,\beta))\mapsto((u,\alpha),(v,\beta))=
(u(x,\alpha, y,\beta),v(x,\alpha, y,\beta))
\end{equation}
where $x, \ y $ belong to a set $\mathcal{X}$, from our consideration the set $\mathcal{X}$ has the structure 
of an algebraic variety,  and the parameters $\alpha, \beta \in
\mathbb{C}^m$. We usually keep the parameters separately and denote
$S((x,\alpha),(y,\beta))$ by $S_{\alpha,\beta}(x,y)$.
We can represent $S_{\alpha,\beta}(x,y)$ as a map assigned to the
edges of an elementary quadrilateral as in Fig.\ref{fig:map}.
\\
\begin{figure}[h]
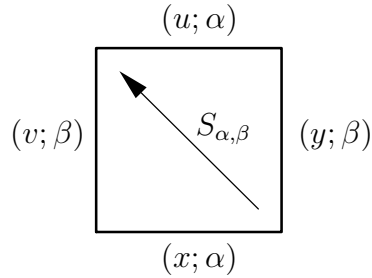

\centertexdraw{ \setunitscale 0.12 \linewd 0.01
\move (7 1) \linewd 0.05 \arrowheadtype t:F \avec(1 7) \lpatt( )
\move (0 0) \linewd 0.1 \lvec(8 0)  \lvec (8 8)  \lvec (0 8)  \lvec (0 0) \lpatt( )
\htext (2.8 -1.9){$(x;\alpha)$} \htext (8.71 3.5){$(y;\beta)$}
\htext (2.8 8.5){$(u;\alpha)$} \htext (-3.7 3.5){$(v;\beta)$}
\htext (4.3 3.8){$S_{\alpha,\beta}$}
}
\caption{A map assigned to the edges of a quadrilateral}
\label{fig:map}
\end{figure}
\begin{definition}
The ordered pair of square matrices $(L(x,\alpha,\zeta) ,  M(x, \alpha, \zeta))$ that depends on a point $x \in \mathcal{X}$, on a parameter
$\alpha \in \mathbb{C}^m$ and on a spectral parameter $\zeta \in \mathbb{C}$ is called a {\em Lax pair}
for the parametric map $S_{\alpha,\beta}$, if
\begin{equation} \label{laxmat}
L(u,\alpha,\zeta)M(v,\beta,\zeta)=M(y,\beta,\zeta)L(x,\alpha,\zeta),
\end{equation}
for any $\zeta\in \mathbb{C}$. Furthermore if equation
is equivalent to $(u,\ v)=S_{\alpha,\beta}(x,y)$ then
we will call $(L(x,\alpha,\zeta) ,  M(x, \alpha, \zeta))$ a {\em strong Lax pair}.
\end{definition}
We usually omit the spectral parameter $\zeta$ and denote the
matrices $L(x,\alpha,\zeta)$ and $M(x, \alpha, \zeta)$ by
$L(x;\alpha)$ and $M(x;\alpha)$ respectively.
\begin{unnumrem}
It is instructive to think of a Lax pair as functions $L,M:\mathcal{X} \times \mathbb{C}^m
\rightarrow Mat({k \times k})$. 
In this consideration if $L=M$ the definition coincides with the definition of the Lax matrix 
which is given in \cite{ves4}. The aim of this work is to study the case where $L \neq M$.   
\end{unnumrem}
Let $S_{\alpha,\beta}$ be a map that admits the Lax pair
$L(x;\alpha) , \ M(x; \alpha)$. Next we consider the standard
periodic `staircase' initial value problem, as in \cite{pnc}, for
integrable lattices difference equations.  Initial values
$x_1,...x_n$ and $y_1,...,y_n$ are assigned to the edges of a
`staircase' on a quadrilateral lattice as in Fig. \ref{fig:mapping},
with periodic boundary conditions $x_{n+1}=x_1, \ y_{n+1}=y_1$. The
edge with the $x_i$ value carries the parameter $\alpha_i$, while 
the one with the $y_i$ value the parameter $\beta_i$, for
$i=1,...,n$. Now, having in mind the representation of
$S_{\alpha_i,\beta_i}$ to the edges of a quadrilateral as in Fig.
\ref{fig:map}, we can compute the values of the next level of the
lattice according to Fig. \ref{fig:mapping}. By $(x_i',y_i')$ we
denote the values $S_{\alpha_i,\beta_i}(x_i,y_i)$ for $i=1,...,n$,
while $k$-primed variables $x_i^{(k)}$ and $y_i^{(k)}$ denote the
corresponding values:
$$(x_i^{(k)},y_{j}^{(k)})=S_{\alpha_i,\beta_{j}}(x_i^{(k-1)},y_{j}^{(k-1)})$$
with $j \equiv i+k-1 (mod n)$. 

\bigskip 

\begin{figure}[h]
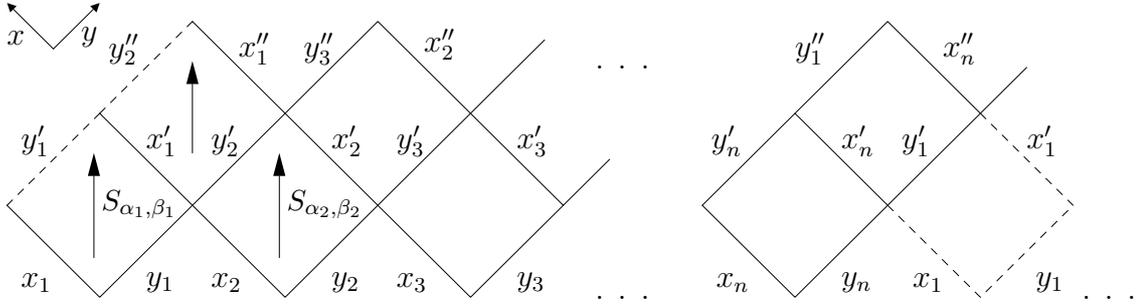

\centertexdraw{ \setunitscale 0.48 \linewd 0.01

\arrowheadsize l:0.1 w:0.1 \arrowheadtype t:F
\move (0.5 1,7) \avec (0 2.2)
\move (0.5 1,7) \avec (1 2.2)

\htext (0 1,76){$x$} \htext (0.8 1.74){$y$}
\move (0 0)  \lvec(1 -1)  \lvec (2 0)  \lvec (3 -1)  \lvec (4 0) \lvec (5 -1)  \lvec (6 0)
\lpatt( )

\htext (6.35 -1){. . .}

\move (7.5 0)  \lvec(8.5 -1)  \lvec (9.5 0) \lpatt(0.08) \lvec (10.5
-1) \lvec (11.5 0) \lpatt( ) \htext (11.6 -1){. . . }

\htext (0.15 -0.93){$x_1$} \htext (1.5 -0.93){$y_1$}\htext (2.2 -0.93){$x_2$} \htext (3.5 -0.93){$y_2$}
\htext (4.2 -0.93){$x_3$} \htext (5.5 -0.93){$y_3$}
\htext (7.65 -0.93){$x_n$} \htext (9 -0.93){$y_n$}\htext (9.77 -0.93){$x_1$} \htext (11.1 -0.93){$y_1$}

\lpatt(0.08) \move (0 0)  \lvec(1 1) \lpatt() \lvec (2 0)  \lvec (3
1)  \lvec (4 0) \lvec (5 1)  \lvec (6 0) \lvec (6.5 0.5)

\move (7.5 0)  \lvec(8.5 1)  \lvec (9.5 0)  \lvec (10.5 1)
\lpatt(0.08) \lvec (11.5 0) \lpatt( )

\lpatt(0.08) \move (1 1)  \lvec(2 2) \lpatt() \lvec (3 1)  \lvec (4
2)  \lvec (5 1) \lvec (5.8 1.8) \htext (6.35 1.5){. . . }

\move (8.5 1)  \lvec(9.5 2)  \lvec (10.5 1) \lvec (11 1.5)

\htext (0.15 0.55){$y_1'$} \htext (1.5 0.55){$x_1'$}\htext (2.2
0.55){$y_2'$} \htext (3.5 0.55){$x_2'$} \htext (4.2 0.55){$y_3'$}
\htext (5.5 0.55){$x_3'$} \htext (7.6 0.55){$y_n'$} \htext (9
0.55){$x_n'$}\htext (9.65 0.55){$y_1'$} \htext (11 0.55){$x_1'$}

\htext (1.1 1.54){$y_2''$}\htext (2.5 1.55){$x_1''$} \htext (3.2
1.55){$y_3''$} \htext (4.5 1.6){$x_2''$}  \htext (8.5
1.55){$y_1''$}\htext (10.1 1.55){$x_n''$}

\arrowheadsize l:0.25 w:0.13 \arrowheadtype t:F
 \move(0.935 -0.57) \avec(0.935 0.57)
\htext (1,02 -0.13) {\small{$S_{\alpha_1,\beta_1} $}} \move(2.935 -0.57) \avec(2.935 0.57) \htext (3.02
-0.13) {\small{$S_{\alpha_2,\beta_2}$}} \move(2 0.57) \avec(2 1.57)

}
 \caption{n-period mapping} \label{fig:mapping}
\end{figure}

\bigskip

For any $n$-periodic `staircase' initial value problem we define the 
{\em `transfer' map}:
\begin{equation} \label{transfer}
T_n:(x_1,...,x_n,y_1,...,y_n)\mapsto (x'_1,...,x'_{n},y'_{2},...,y'_n,y'_1)
\end{equation}
and the {\em k-`transfer' map}:
\begin{equation} \label{ktransfer}
T^{k}_n:(x_1,...,x_n,y_1,...,y_n)\mapsto (x^{(k)}_1,...,x^{(k)}_{n},
\underbrace{y^{(k)}_{d+1},...,y^{(k)}_n}_{n-d},\underbrace{y^{(k)}_1,...,y^{(k)}_d}_{d}), 
\end{equation}
with $d \equiv k (modn)$ and $T^{1}_n=T_n$. 
We observe that 
$$T_n^n(x_1,...,x_n,y_1,...,y_n)=(x_1^{(n)},...,x_n^{(n)},y_1^{(n)},...,y_n^{(n)}).$$
We also define the {\em monodromy matrix}
$M_n(x_1,...,x_n,y_1,...,y_n)= \overset{\curvearrowleft
}{\prod\limits_{i=1}^{n}}M(y_{i};\beta_i)L(x_{i};\alpha_i ),$ where
${\curvearrowleft }$ indicates that the elements
$M(y_{i};\beta)L(x_{i};\alpha )$ in the product are arranged from
right to left.

So, for example, for the 1-periodic initial value problem on the
lattice the transfer map will be
$T_1(x,y)=(x',y')=S_{\alpha,\beta}(x,y),$ with corresponding
monodromy matrix $M_1(x,y)=M(y;\beta)L(x;\alpha )$, while for the
2-periodic case, $$T_2(x_1,x_2,y_1,y_2)=(x'_1,x'_2,y'_2,y'_1),$$ where
$x'_i$ and $y'_i$ are given by:
$(x'_i,y'_i)=S_{\alpha_i,\beta_i}(x_i,y_i)$, for $i=1,2$. The
monodromy matrix in this case is
$M_2(x_1,x_2,y_1,y_2)=M(y_2;\beta_2)L(x_2;\alpha_2
)M(y_1;\beta_1)L(x_1;\alpha_1 )$.

The equivalent proposition to the one in \cite{ves2} for the
$n$-periodic `staircase' initial value problem holds.
\begin{proposition} \label{transfer}
The transfer map preserves the spectrum of the monodromy matrix.
\end{proposition}

\begin{proof} \label{integ}
The definition of the transfer map and the monodromy matrix implies:
$$M_n(T_n(x_1,...,x_n,y_1,...,y_n))M(y'_1;\beta_1)=M(y'_1;\beta_1)
\overset{\curvearrowleft
}{\prod\limits_{i=1}^{n}}L(x'_{i};\alpha_i)M(y'_{i};\beta_i ).$$
Since $L(x;\alpha)$, $M(x;\alpha)$ is a Lax pair for the map and
$(x_i',y_i')=S_{\alpha_i,\beta_i}(x_i,y_i)$ we have that
$L(x'_{i};\alpha_i)M(y'_{i};\beta_i
)=M(y_{i};\beta_i)L(x_{i};\alpha_i )$, so
\begin{eqnarray*}
M_n(T_n(x_1,...,x_n,y_1,...,y_n))M(y'_1;\beta_1)&=& M(y'_1;\beta_1)
\overset{\curvearrowleft
}{\prod\limits_{i=1}^{n}}M(y_{i};\beta_i)L(x_{i};\alpha_i ) \\ &=&
M(y'_1;\beta_1)M_n(x_1,...,x_n,y_1,...,y_n),
\end{eqnarray*}
or $
M_n(T_n(x_1,...,x_n,y_1,...,y_n))=M(y'_1;\beta_1)M_n(x_1,...,x_n,y_1,...,y_n)M^{-1}(y'_1;\beta_1)$.
\end{proof}
Similarly proposition \ref{transfer} holds also for any $k$-transfer map. So
the spectrum of the monodromy matrix gives integrals of the transfer map. If we derive 
$N$ functionally independent integrals, 
the transfer map is integrable provided that 
a symplectic
structure $\omega$ exists in a 2N-dimensional phase space such that: 
i) the transfer map is symplectic with respect to $\omega$ and ii) the integrals are
in involution. In the case of polynomial Lax matrices that we are dealing bellow,
the Sklyanin bracket (\ref{sklyanin}) 
give rise to a symplectic structure after reduction to the symplectic leaves \cite{skly3}.

%%%%%%%%%%%%%%%%%%%%%%%%%%%%%%%%%%%%%%%%%%%%%%%%%%%%%%%%%%%%%%%%%%%%%%%%%%%%%%%%%%%%%%%%%%%%%%%%%%%
%%%%%%%%%%%%%%%%%%%%%%%%%%%%%%%%%%%%%%%%%%%%%%%%%%%%%%%%%%%%%%%%%%%%%%%%%%%%%%%%%%%%%%%%%%%%%%%%%%%
%%%%%%%%%%%%%%%%%%%%%%%%%%%%%%%%%%%%%%%%%%%%%%%%%%%%%%%%%%%%%%%%%%%%%%%%%%%%%%%%%%%%%%%%%%%%%%%%%%%
\section{Entwining Yang--Baxter equation and Lax triples} \label{sec3-d}

Now we consider three parametric maps of the form (\ref{mapphi}), 
$S_{\alpha,\beta}, \ R_{\alpha,\beta}, \ T_{\alpha,\beta}:
\mathcal{X} \times \mathcal{X} \rightarrow \mathcal{X} \times \mathcal{X}$. 
In correspondence with \cite{Nich} we call the equation 
\begin{equation} \label{eYBeq} 
T_{\beta,\gamma}^{23}\circ R_{\alpha,\gamma}^{13}\circ S_{\alpha,\beta}^{12}=
S_{\alpha,\beta}^{12}\circ R_{\alpha,\gamma}^{13}\circ T_{\beta,\gamma}^{23}, 
\end{equation} 
\textit{the entwining quantum Yang-Baxter equation} or just the entwining YB equation. 
Here by $S^{ij}$ (respectively $T^{ij}$ and $R^{ij}$) 
for $i,j=1,2,3$, we denote the map that acts as $S$ (resp. $T$ and $R$) 
on the $i$ and 
$j$ factor of $\mathcal{X} \times \mathcal{X} \times \mathcal{X}$
and identically on the others.  

A \textit{Lax triple} (resp. \textit{strong Lax triple}) 
of three maps $(S_{\alpha,\beta}, R_{\alpha,\beta},  T_{\alpha,\beta})$ 
is a triple of matrices $$(L_1(x;\alpha),  L_2(x;\alpha), L_3(x;\alpha)),$$ 
such that the pairs $(L_1(x;\alpha),  L_2(x;\alpha))$, $(L_1(x;\alpha),  L_3(x;\alpha))$, 
$(L_2(x;\alpha),  L_3(x;\alpha))$ are Lax pairs (resp. strong Lax pairs) of the maps 
$S_{\alpha,\beta}, \ R_{\alpha,\beta}, \ T_{\alpha,\beta}$ respectively. 

\begin{proposition} \label{3product}
Let 
$(S_{\alpha,\beta},  R_{\alpha,\beta},  T_{\alpha,\beta})$ be 
three maps on $\mathcal{X} \times \mathcal{X}$ that admit the Lax triple   
$(L_1(x;\alpha),  L_2(x;\alpha), L_3(x;\alpha)).$ 
If the equation 
\begin{equation}
L_1( \hat{x}; \alpha )L_2( \hat{y}; \beta )L_3(\hat{z}; \gamma )=
L_1(x; \alpha )L_2(y; \beta )L_3(z; \gamma ) \label{xyz}
\end{equation}
implies that $\hat{x}=x, \ \hat{y}=y$ and $\hat{z}=z$, then 
$S_{\alpha,\beta}, \ R_{\alpha,\beta}, \ T_{\alpha,\beta}$ 
satisfy the entwining YB equation (\ref{eYBeq}).
\end{proposition}

\begin{proof}
Let 
\begin{eqnarray*}
S_{\alpha,\beta}^{12}(x,y,z)&=&(x',y',z), \\
R_{\alpha,\gamma}^{13}\circ S_{\alpha,\beta}^{12}(x,y,z)&=&(x'',y',z'), \\ 
T_{\beta,\gamma}^{23}\circ R_{\alpha,\gamma}^{13}\circ S_{\alpha,\beta}^{12}(x,y,z)&=&
(x'',y'',z'')
\end{eqnarray*}
We can represent these maps 
as chains of maps at the down, back, left faces of a cube as
in $(i)$ of Fig.\ref{fig:3d}. All the parallel edges to the $x$ (resp.
$y,z$) axis carry the parameter $\alpha$ (resp. $\beta$, $\gamma$).
From the corresponding Lax pairs we have 
$L_2(y;\beta)L_1(x;\alpha)=L_1(x';\alpha)L_2(y';\beta)$, so
$L_3(z;\gamma)L_2(y;\beta)L_1(x;\alpha)=(L_3(z;\gamma)L_1(x';\alpha))L_2(y';\beta)
=L_1(x'';\alpha)(L_3(z';\gamma)L_2(y';\beta))
=L_1(x'';\alpha)L_2(y''\beta)L_3(z'';\gamma)$. 

We also assume that 
\begin{eqnarray*}
T_{\beta,\gamma}^{23}(x,y,z)&=&(x,\tilde{y},\tilde{z}), \\
R_{\alpha,\gamma}^{13}\circ T_{\beta,\gamma}^{23}(x,y,z)&=&(\tilde{x},\tilde{y},\tilde{\tilde{z}}), \\ 
S_{\alpha,\beta}^{12}\circ R_{\alpha,\gamma}^{13}\circ T_{\beta,\gamma}^{23}(x,y,z)&=&
(\tilde{\tilde{x}},\tilde{\tilde{y}},\tilde{\tilde{z}})
\end{eqnarray*}
These maps are represented at right, front and upper faces of the cube as in 
$(ii)$ of Fig.\ref{fig:3d}. 
Similarly from the Lax pairs we get
$L_3(z;\gamma)L_2(y;\beta)L_1(x;\alpha)=
L_1(\tilde{\tilde{x}};\alpha)L_2(\tilde{\tilde{y}};\beta)L_3(\tilde{\tilde{z}};\gamma)$.
So finally we have that 
$$L_1(x'';\alpha)L_2(y''\beta)L_3(z'';\gamma)=
L_1(\tilde{\tilde{x}};\alpha)L_2(\tilde{\tilde{y}};\beta)L_3(\tilde{\tilde{z}};\gamma)$$
which implies (from the assumptions of the proposition) that $x''= \tilde{\tilde{x}},  \
y''=\tilde{\tilde{y}}, \ z''=\tilde{\tilde{z}}$. i.e. 
$T_{\beta,\gamma}^{23}\circ R_{\alpha,\gamma}^{13}\circ S_{\alpha,\beta}^{12}=
S_{\alpha,\beta}^{12}\circ R_{\alpha,\gamma}^{13}\circ T_{\beta,\gamma}^{23}$.
\end{proof}

\begin{figure}[h]
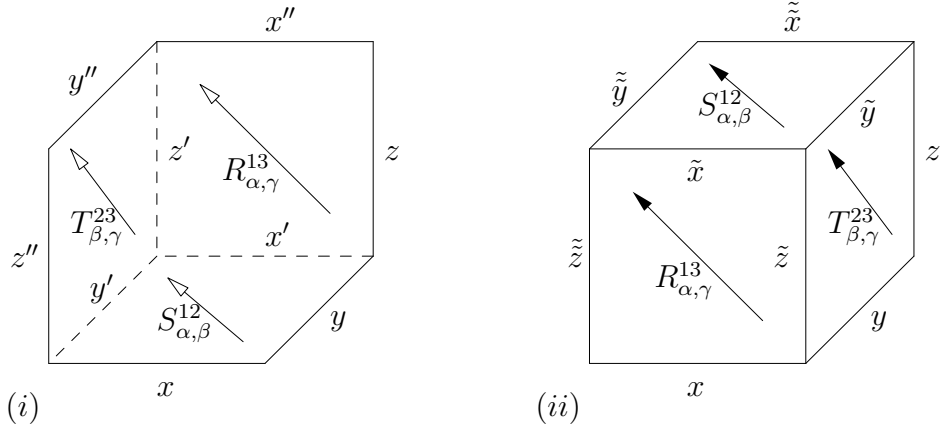

\centertexdraw{ \setunitscale 1.12 \linewd 0.06

\arrowheadsize l:0.1 w:0.05

\move (1. 0)  \linewd 0.001 \lvec(0.5 -0.5) \lpatt()
\lvec(0.5 0.5)\lpatt() \lvec(1.0 1.0) \lpatt()
\move(0.5 0.5) \lpatt()\lvec(-0.5  0.5)\lpatt() \lvec(0.0 1.0)\lpatt(0.05 0.05)

\move(-2.5 0.0) \lpatt(0.05 0.05)
\lvec(-3.0 -0.5) \lpatt()
\lvec(-2.0 -0.5) \lpatt()
\lvec(-1.5 0.0) \lpatt(0.05 0.05)
\lvec(-2.5 -0.0)  \lpatt()
\move(-1.5 0.0) \lvec(-1.5 1.0)\lpatt() \lvec(-2.5 1.0)
\lpatt(0.05 0.05) \lvec(-2.5 0.0) \lpatt()
\move(-3.0 -0.5) \lvec(-3.0 0.5) \lpatt()

\lvec(-2.5 1.0) \lpatt()
\move(0 1) \lvec(1 1) \lpatt()\lvec(1 0)\lpatt()
\move(-0.5 0.5) \lvec(-0.5 -0.5) \lpatt() \lvec(0.5 -0.5)\lpatt()

\move(-2.6 0.1) \avec(-2.9 0.5)\lpatt()
\move(-1.7 0.2) \avec(-2.3 0.8)\lpatt()
\move(-2.1 -0.4) \avec(-2.45 -0.1)\lpatt()
\arrowheadtype t:F
\move(0.9 0.1) \avec(0.6 0.5)\lpatt()
\move(0.3 -0.3) \avec(-0.3 0.3)\lpatt()
\move(0.4 0.6) \avec(0.05 0.9)\lpatt()

\htext (-0.05 -0.65) {$x$}
\htext (0.8 -0.35) {$y$}
\htext (1.05 0.45) {$z$}
\htext (0.35 -0.05) {$\tilde{z}$}
\htext (0.75 0.6) {$\tilde{y}$}
\htext (-0.05 0.35) {$\tilde{x}$}
\htext (-0.6 -0.05) {$\tilde{\tilde{z}}$}
\htext (-0.4 0.7) {$\tilde{\tilde{y}}$}
\htext (0.4 1.05) {$\tilde{\tilde{x}}$}

\htext (-2.5 -0.65) {$x$}
\htext (-1.7 -0.35) {$y$}
\htext (-1.45 0.45) {$z$}
\htext (-2.0 0.05) {$x'$}
\htext (-2.81 -0.22) {$y'$}
\htext (-2.0 1.05) {$x''$}
\htext (-2.45 0.45) {$z'$}
\htext (-2.91 0.75) {$y''$}
\htext (-3.18 -0.05) {$z''$}

\htext (0.6 0.05) {$T^{23}_{\beta,\gamma}$}
\htext (-0.2 -0.2) {$R^{13}_{\alpha,\gamma}$}
\htext (0.0 0.6) {$S_{\alpha,\beta}^{12}$}

\htext (-2.5 -0.4) {$S_{\alpha,\beta}^{12}$}
\htext (-2.2 0.3) {$R^{13}_{\alpha,\gamma}$}
\htext (-2.9 0.05) {$T^{23}_{\beta,\gamma}$}

\htext (-3.2 -0.78)  {$(i)$}
\htext (-0.75 -0.78) {$(ii)$}

}
\caption{Cubic representation of the 3-d compatibility condition}
\label{fig:3d}
\end{figure}

\newpage 

\begin{corollary} \label{corYB}
If $(L_1(x;\alpha),  L_2(x;\alpha), L_3(x;\alpha))$ is a strong Lax triple that satisfies  
Prop. \ref{3product} and 
$L_1(x;\alpha)= L_2(x; \alpha)=L_3(x; \alpha)$, then 
$S_{\alpha,\beta}:\mathcal{X} \times \mathcal{X} \rightarrow \mathcal{X} \times \mathcal{X}$  
is a parametric Yang-Baxter map.  
\end{corollary}

\begin{proof}
Since the Lax pairs $(L_1(x;\alpha),  L_2(x;\alpha))$, $(L_1(x;\alpha),  L_3(x;\alpha))$, 
$(L_2(x;\alpha),  L_3(x;\alpha))$ of $S_{\alpha,\beta}, \ R_{\alpha,\beta}$ and $T_{\alpha,\beta}$ 
are a strong, then for $L_1=L_2=L_3$,    
$S_{\alpha,\beta}= R_{\alpha,\beta}=T_{\alpha,\beta}$. So equation 
(\ref{eYBeq}) becomes the Yang-Baxter equation 
\begin{equation} \label{YB}
S_{\beta,\gamma}^{23}\circ S_{\alpha,\gamma}^{13}\circ S_{\alpha,\beta}^{12}=
S_{\alpha,\beta}^{12}\circ S_{\alpha,\gamma}^{13}\circ S_{\beta,\gamma}^{23}.
\end{equation}
\end{proof}
In this case the matrix $L_1(x;\alpha)=L_2 (x; \alpha)=L_3(x; \alpha)$ is called a (strong) 
Lax matrix of the Yang-Baxter map 
$S_{\alpha,\beta}$ \cite{ves4}.

%%%%%%%%%%%%%%%%%%%%%%%%%%%%%%%%%%%%%%%%%%%%%%%%%%%%%%%%%%%%%%%%%%%%%%%%%%%%%%%%%%%%%%
%%%%%%%%%%%%%%%%%%%%%%%%%%%%%%%%%%%%%%%%%%%%%%%%%%%%%%%%%%%%%%%%%%%%%%%%%%%%%%%%%%%%%%
%%%%%%%%%%%%%%%%%%%%%%%%%%%%%%%%%%%%%%%%%%%%%%%%%%%%%%%%%%%%%%%%%%%%%%%%%%%%%%%%%%%%%%
\section{Binomial Lax pairs and triples} \label{secbin}
A matrix re-factorization procedure provides a way of constructing strong Lax pairs and  
consequently Lax triples. 
Following the lines of \cite{kp1,ves1} we consider the set $\mathcal{L}^2$ of first degree 
$2 \times 2$ polynomial 
matrices $X-\zeta A$, and we denote by $p_{X}^A$ the polynomial 
$$p_{X}^A(\zeta):=\det(X-\zeta A)=
f_{2}(X;A) \zeta^{2}-f_{1}(X;A) \zeta + f_{0}(X;A).$$ 
For $X=[x_{ij}]$ and $A=[\alpha_{ij}]$, 
$$f_{0}(X;A)= \det X, \   
f_{1}(X;A)=a_{22}x_{11}-a_{21}x_{12}-a_{12}x_{21}+a_{11}x_{22}, \  
f_{2}(X;A)= \det A.$$  
We also consider the matrix functions functions $\Pi_1,\ \Pi_2,$ with
\begin{eqnarray} \label{p1p2}
\Pi _{1}(X,Y)&=& f_{2}(X;A)(YA+BX)-f_{1}(X;A)AB, \\
\Pi _{2}(X,Y)&=& f_{2}(X;A)YX-f_{0}(X;A)AB .
\end{eqnarray}

\begin{proposition} \label{UV} 
Let $A, \ B$ be invertible $2 \times 2$ matrices, such that $AB=BA$ and 
$X, Y \in Mat( 2\times 2)$ with $\det \Pi _{1}(X,Y) \neq 0$ . Then 
\begin{equation}
(U-\zeta A)(V-\zeta B)= (Y- \zeta B)(X-\zeta A), \label{fact}
\end{equation}
and 
$p_{U}^{A}(\zeta)=p_{X}^{A}(\zeta)$ 
(equivalently $p_{V}^{B}(\zeta)=p_{Y}^{B}(\zeta)$), 
iff 
\begin{eqnarray}  \label{U}
U=U(X,Y) &:=& \Pi _{2}(X,Y) \Pi _{1}(X,Y)^{-1}A, \\  
V=V(X,Y) &:=& A^{-1}(YA+BX-U(X,Y)B). \label{V}
\end{eqnarray}
\end{proposition}
The proof is given in Appendix A. For $A=B$, the proof has appeared in \cite{kp1}.  
The fact that $p_{U}^{A}(\zeta)=p_{X}^{A}(\zeta)$ and
$p_{V}^{B}(\zeta)=p_{Y}^{B}(\zeta)$, or equivalently 
$$f_{i}(U;A)=f_{i}(X;A), \ f_{i}(V;B)=f_{i}(Y;B), \ \text{for} \ i=0,1,2, $$ 
is crucial and leads to the construction of Lax pairs and symplectic solutions of the entwining YB equation. 

A Poisson structure on $\mathcal{L}^2$ is defined by the Sklyanin bracket \cite{skly2}: 
\begin{equation} 
\{L(\zeta ) \ \overset{\otimes }{,} \ L(\eta)\}=[\frac{r}{\zeta
-\eta},L(\zeta )\otimes L(\eta)],  \label{sklyanin}
\end{equation} 
where $L(\zeta)=X-\zeta A$ and $r$ denotes the permutation matrix,
 $r(x\otimes y) = y\otimes x$. 
There are six Casimir functions on $\mathcal{L}^2$ which are  
the elements $a_{ij}$ of the matrix $A$ and 
the coefficients of the polynomial $p_{X}^A(\zeta)$, 
$f_{0}(X;A), \ f_{1}(X;A)$.
For any constant matrix $A$ we denote by $\mathcal{L}^2_{A}$ the level set
$\mathcal{L}^2_{A}=\{X-\zeta A \ / \ X \in Mat(2\times2) \}$.

By a direct computation we can prove that the map $(X,Y)\mapsto (U,V)$, for $U$, $V$,  
(\ref{U}) and (\ref{V}) respectively, is a Poisson 
map $\mathcal{L}^2\times \mathcal{L}^2
\rightarrow  \mathcal{L}^2\times \mathcal{L}^2$ (we extend the Sklyanin bracket to the cartesian 
product  $\mathcal{L}^2 \times \mathcal{L}^2$ in the natural way).

Let $A_1,A_2,A_3$ be three invertible matrices that 
commute with each other. 
We restrict to a level set on $\mathcal{L}^2_{A_1}$, of the Casimir functions $f_0$ and $f_1$,  
by solving the system $f_{0}(X;{A_1})=\alpha_{0}$, $f_{1}(X;{A_1})=\alpha_{1}$,  
with respect to two elements of $X$. We denote the two remaining elements of $X$ 
by $x_1$ and $x_2$. 
In this way we define the matrix $L_1'(x_{1},x_{2};\bar{\alpha})$, with 
$\bar{\alpha}=(\alpha_0,\alpha_1)$, such that  
\begin{equation*}
f_{0}(L_1'(x_{1},x_{2};\bar{\alpha});A_1) = \alpha_{0} \ \text{and} \   
f_{1}(L_1'(x_{1},x_{2};\bar{\alpha});A_1) = \alpha_{1}  
\end{equation*}
and the two dimensional symplectic leaves of  
$\mathcal{L}^2_{A_1}$ 
$$\Sigma_{A_1}(\bar{\alpha})=\{L'(x_{1},x_{2};\bar{\alpha}) -\zeta A_1 \ 
/ \ x_1,x_2 \in I\subset \mathbb{C} \},$$ 
with respect to the reduced Sklyanin structure (\ref{sklyanin}).  
In a similar way we define the matrices $L_2'(x_{1},x_{2};\bar{\alpha})$ and 
$L_3'(x_{1},x_{2};\bar{\alpha})$ 
from the restriction on the level sets of the Casimir functions on 
$\mathcal{L}^2_{A_2}$ and $\mathcal{L}^2_{A_3}$ respectively i.e. 
\begin{eqnarray*}
f_{0}(L_2'(x_{1},x_{2};\bar{\alpha});A_2) = \alpha_{0},  \   
f_{1}(L_2'(x_{1},x_{2};\bar{\alpha});A_2) = \alpha_{1},  \\
f_{0}(L_3'(x_{1},x_{2};\bar{\alpha});A_3) = \alpha_{0},  \   
f_{1}(L_3'(x_{1},x_{2};\bar{\alpha});A_3) = \alpha_{1} \ 
\end{eqnarray*}
and the corresponding symplectic leaves $\Sigma_{A_2}(\bar{\alpha})$ and $\Sigma_{A_3}(\bar{\alpha})$. 
\begin{theorem} \label{lpun}
The equations 
\begin{equation} \label{uniq} 
(L_i'(u^{ij}_{1},u^{ij}_{2};\bar{\alpha})-\zeta A_i) (L_j'(v^{ij}_{1},v^{ij}_{2};\bar{\beta})-\zeta A_j)  
=
(L_j'(y_{1},y_{2};\bar{\beta})-\zeta A_j) (L_i'(x_{1},x_{2};\bar{\alpha})-\zeta A_i),  
\end{equation}
for $i,j=1,2,3$, are uniquely solvable with respect to $u^{ij}_1,\ u^{ij}_2, \ v^{ij}_1$ and $v^{ij}_2$. 
The parametric maps $(S_{\bar{\alpha},\bar{\beta}},R_{\bar{\alpha},\bar{\beta}},T_{\bar{\alpha},\bar{\beta}})$, 
with 
\begin{eqnarray*} 
S_{\bar{\alpha},\bar{\beta}}:((x_1,x_2),(y_1,y_2)) &\mapsto & ((u^{12}_1,u^{12}_2),(v^{12}_1,v^{12}_2)), \\ 
R_{\bar{\alpha},\bar{\beta}}:((x_1,x_2),(y_1,y_2)) &\mapsto & ((u^{13}_1,u^{13}_2),(v^{13}_1,v^{13}_2)), \\ 
T_{\bar{\alpha},\bar{\beta}}:((x_1,x_2),(y_1,y_2)) &\mapsto & ((u^{23}_1,u^{23}_2),(v^{23}_1,v^{23}_2)) 
\end{eqnarray*}
are symplectic and admit the strong Lax triple 
$(L_1(x_1,x_2,\bar{\alpha}),L_2(x_1,x_2,\bar{\alpha}),L_3(x_1,x_2,\bar{\alpha}))$, 
where $$L_i(x_1,x_2,\bar{\alpha})= L_i'(x_1,x_2,\bar{\alpha})-\zeta A_i, \ \text{for} \ i=1,2,3.$$
Moreover they satisfy the entwining YB equation 
$T_{\beta,\gamma}^{23} R_{\alpha,\gamma}^{13} S_{\alpha,\beta}^{12}=
S_{\alpha,\beta}^{12} R_{\alpha,\gamma}^{13} T_{\beta,\gamma}^{23}$. 
\end{theorem}
The proof of this theorem is given in appendix B. 

As it is remarked in \cite{kp1}, there are cases where limits of quadrirational YB maps give rise to new 
degenerate maps. A similar procedure can be applied here as well. For example we can 
consider commuting invertible matrices $A_i=A_i(\varepsilon)$, depending on a 
parameter $\varepsilon$, such that $\lim_{\varepsilon \rightarrow 0}detA_i(\varepsilon)=0$, 
for some $i \in \{1,2,3\}$, and construct the corresponding maps 
$(S_{\bar{\alpha},\bar{\beta}}(\varepsilon),R_{\bar{\alpha},\bar{\beta}}(\varepsilon),
T_{\bar{\alpha},\bar{\beta}}(\varepsilon))$ 
that we described. A new solution of the entwining YB equation can be derived by taking 
the limit (if it exists) of $(S_{\bar{\alpha},\bar{\beta}}(\varepsilon),R_{\bar{\alpha},\bar{\beta}}(\varepsilon),
T_{\bar{\alpha},\bar{\beta}}(\varepsilon))$ for $\varepsilon \rightarrow 0$. 

%%%%%%%%%%%%%%%%%%%%%%%%%%%%%%%%%%%%%%%%%%%%%%%%%%%%%%%%%%%%%%%%%%%%%%%%%%%%%%%%%%%%%%%%%%%%%%%%%%%%%%%
%%%%%%%%%%%%%%%%%%%%%%%%%%%%%%%%%%%%%%%%%%%%%%%%%%%%%%%%%%%%%%%%%%%%%%%%%%%%%%%%%%%%%%%%%%%%%%%%%%%%%%%
%%%%%%%%%%%%%%%%%%%%%%%%%%%%%%%%%%%%%%%%%%%%%%%%%%%%%%%%%%%%%%%%%%%%%%%%%%%%%%%%%%%%%%%%%%%%%%%%%%%%%%%
\section{Integrable rational maps on $\mathbb{C}^2 \times \mathbb{C}^2$} \label{secent}
We give here an example of two solutions of the entwining YB equation from a pair of 
commuting matrices. 

We consider $X-\zeta A_i \in \mathcal{L}^2_{A_i}$, for $i=1,2,$ 
with   
\begin{equation*} 
X= 
\left(
\begin{array}{ll}
 x_1 & x_2 \\
 x_3 & x_4
\end{array} 
\right) , \  
A_1= 
\left(
\begin{array}{ll}
 1 & 0 \\
 0 & 1
\end{array} 
\right) 
\ \text{and}  \ 
A_2= 
\left(
\begin{array}{ll}
 1 & 0 \\
 0 & \varepsilon
\end{array}
\right) , \ \varepsilon \neq 0. 
\end{equation*}
First we will construct symplectic rational maps on symplectic leaves of 
$\mathcal{L}^2_{A_1} \times \mathcal{L}^2_{A_1}$,  
$\mathcal{L}^2_{A_2} \times \mathcal{L}^2_{A_2}$, 
$\mathcal{L}^2_{A_1} \times \mathcal{L}^2_{A_2}$ and $\mathcal{L}^2_{A_2} \times \mathcal{L}^2_{A_1}$ 
respectively and their corresponding Lax pairs. Next by considering two cases 
$A_3=A_1$ and $A_3=A_2$ we derive two solutions of the entwining YB equation. All these 
maps are integrable. 

%%%%%%%%%%%%%%%%%%%%%%%%%%%%%%%%%%%%%%%%%%%%%%%%%%%%%%%%%%%%%%%%%%%%%%%%%%%%%%%%%%%
\subsection{ A Yang-Baxter map on $\Sigma_{A_1}(\alpha,1) \times \Sigma_{A_1}(\beta,1)$}
The Casimir functions on $\mathcal{L}^2_{A_1}$ are  
$$f_0(X;A_1)=x_1x_4-x_2x_3, \ f_1(X;A_1)=x_1+x_4.$$
We set $f_0(X;A_1)=\alpha$, $f_1(X;A_1)=1$ and solve with 
respect to $x_3$ and $x_4$ to get the matrix 
\begin{equation} 
\label{L1}
L_1'(x_1,x_2;\alpha)=
\left(
\begin{array}{cc}
 {x_1} & {x_2} \\
 -\frac{\alpha+({x_1}-1) {x_1}}{{x_2}} & 1-{x_1}
\end{array}
\right)
\end{equation}
and the Lax matrix 
\begin{equation} \label{laxL1}
L_1(x_1,x_2;\alpha)=L'(x_1,x_2;\alpha)-\zeta A_1.
\end{equation} 
As expected from theorem \ref{lpun} and corollary \ref{corYB}, the equation 
$$L_1(u'_1,u'_2;\alpha)L_1(v'_1,v'_2;\beta)=L_1(y_1,y_2;\beta)L_1(x_1,x_2;\alpha)$$ 
is uniquely solvable with respect to $u'_1,u'_2,v'_1,v'_2$ and implies the parametric 
Yang-Baxter map
\begin{equation} \label{YB1}
R_{\alpha,\beta}((x_1,x_2),(y_1,y_2))=((u'_1,u'_2),(v'_1,v'_2)),
\end{equation}
where 
\begin{eqnarray*}
u'_1 &=&  y_{1}+\frac{x_{2}y_{2}}{N}(\alpha-\beta)  (x_{1}+y_{1}-1)  ,\\
u'_2 &=& \frac{y_{2}}{N} \left(\alpha (x_{2}+y_{2})^2+(x_{2} (y_{1}-1)-x_{1} y_{2}) (x_{2} y_{1}-x_{1}
   y_{2}+y_{2})\right), \\
v'_1 &=& x_{1}-\frac{x_{2}y_{2}}{N}(\alpha-\beta)  (x_{1}+y_{1}-1) , \\ 
v'_2 &=& \frac{x_{2}}{N} \left(\beta (x_{2}+y_{2})^2+(x_{2} (y_{1}-1)-x_{1} y_{2}) (x_{2} y_{1}-x_{1}
   y_{2}+y_{2})\right), \\ 
   N &=& \alpha {y_2}^2 + \beta {x_2}^2+(\alpha+\beta -1) {y_2} {x_2}+(x_2 y_1 - x_1 y_2)
   (x_2 y_1 - x_1 y_2+y_2-x_2).
\end{eqnarray*}
We can verify that $u'_1=U'_{11}$, $u'_2=U'_{12}, \ v'_1=V'_{11}$ and $v'_2=V'_{12}$, where $U'_{ij}$ and $V'_{ij}$ 
are the $ij$ elements of the matrices 
$$U':=U(L'(x_{1},x_{2};\alpha),L'(y_1,y_2,\beta)), \  V':=V(L'(x_{1},x_{2};\alpha),L'(y_1,y_2,\beta))$$ 
defined by (\ref{U}) and (\ref{V}) respectively. 
 
%%%%%%%%%%%%%%%%%%%%%%%%%%%%%%%%%%%%%%%%%%%%%%%%%%%%%%%%%%%%%%%%%%%%%%%%%%%%%%%%%%%
%%%%%%%%%%%%%%%%%%%%%%%%%%%%%%%%%%%%%%%%%%%%%%%%%%%%%%%%%%%%%%%%%%%%%%%%%%%%%%%%%%%%%%
\subsection{ A Yang-Baxter map on $\Sigma_{A_2}(\alpha,1) \times \Sigma_{A_2}(\beta,1)$}
The construction of the following Yang-Baxter map appears in \cite{kp1}. 
The Casimir functions on $\mathcal{L}^2_{A_2}$ are 
$$f_0(X;A_2)=x_1 x_4-x_2 x_3, \ f_1(X;A_2)=\varepsilon x_1+x_4.$$
We set here as well $f_0(X;A_2)=\alpha$, $f_1(X;A_2)=1$ and solve 
with respect to $x_1,\ x_4$. By performing the change of 
variables: $x_2 \mapsto x_1$ and $x_3 \mapsto x_2$ we derive 
the matrix 
\begin{equation} \label{Me}
M_{\varepsilon}'(x_1,x_2;\alpha)=\left(
\begin{array}{cc}
 \frac{1-(1-4 \varepsilon (\alpha+{x_1} {x_2}))^{1/2}}{2 \varepsilon} & {x_1} \\
 {x_2} & \frac{1}{2} (1-4 \varepsilon (\alpha+{x_1} {x_2}))^{1/2}+\frac{1}{2}
 \end{array} \right)
\end{equation}
and the corresponding Lax matrix $M_{\varepsilon}(x_1,x_2;\alpha)=M_{\varepsilon}'(x_1,x_2;\alpha)-\zeta A_2$.
In this way we obtain the non-degenerate Yang-Baxter map 
\begin{equation} \label{YB2}
R_{\alpha,\beta}^{\varepsilon}((x_1,x_2),(y_1,y_2))=((\mathbf{u}_1,\mathbf{u}_2),(\mathbf{v}_1,\mathbf{v}_2)),
\end{equation}
with  $\mathbf{u}_1=\mathbf{U}_{12}$, $\mathbf{u}_2=\mathbf{U}_{21}, 
\ \mathbf{v}_1=\mathbf{V}_{12}$, $\mathbf{v}_2=\mathbf{V}_{21}$ the corresponding elements of 
$\mathbf{U}=U(M_{\varepsilon}'(x_1,x_2;\alpha),M_{\varepsilon}'(y_1,y_2;\beta))$,    
$\mathbf{V}=V(M_{\varepsilon}'(x_1,x_2;\alpha),M_{\varepsilon}'(y_1,y_2;\beta))$. As in \cite{kp1} we 
take the limit of (\ref{YB2}), for $\varepsilon \rightarrow 0$, in order to derive the degenerate Yang-Baxter map
\begin{equation} \label{YB3}
\bar{R}_{\alpha,\beta}((x_1,x_2),(y_1,y_2))=
\lim_{\varepsilon \rightarrow 0}R_{\alpha,\beta}^{\varepsilon}((x_1,x_2),(y_1,y_2))
=((\bar{u}_1,\bar{u}_2),(\bar{v}_1,\bar{v}_2))
\end{equation}
with 
\begin{eqnarray*}
\bar{u}_{1} =  y_{1}-\frac{(\alpha-\beta)x_{1}}
{1+x_{1}y_{2}}, \ 
\bar{u}_{2}= y_{2}, \ 
\bar{v}_{1} = x_{1}, \  
\bar{v}_{2} = x_{2}+\frac{(\alpha-\beta)y_{2}}{1+x_{1}y_{2}},
\end{eqnarray*}
and strong Lax matrix
\begin{equation} \label{LaxL2}
L_2(x_1,x_2;\alpha)=\lim_{\varepsilon \rightarrow 0} M_{\varepsilon}(x_1,x_2;\alpha)=
\left(
\begin{array}{cc}
 {x_1} x_2+\alpha-\zeta & {x_1} \\
 x_2 & 1
\end{array}
\right).
\end{equation} 
%%%%%%%%%%%%%%%%%%%%%%%%%%%%%%%%%%%%%%%%%%%%%%%%%%%%%%%%%%%%%%%%%%%%%%%%%%%%%%%%%%%%%%%%%%%%%%%%%%%%%%%
%%%%%%%%%%%%%%%%%%%%%%%%%%%%%%%%%%%%%%%%%%%%%%%%%%%%%%%%%%%%%%%%%%%%%%%%%%%%%%%%%%%%%%%%%%%%%%%%%%%%%%%
\subsection{The maps on $\Sigma_{A_1}(\alpha,1) \times \Sigma_{A_2}(\beta,1)$ 
and on $\Sigma_{A_2}(\alpha,1) \times \Sigma_{A_1}(\beta,1)$} 
The equation 
$$L_1(u_1,u_2;\alpha)L_2(v_1,v_2;\beta)=L_2(y_1,y_2;\beta)L_1(x_1,x_2;\alpha)$$
with $L_1, \ L_2$ (\ref{laxL1}) and (\ref{LaxL2}) respectively, yields the unique solution 
\begin{eqnarray*}
u_1 &=& {y_2} {y_1}-\frac{{x_1} {y_1}}{{x_2}}+\frac{(\alpha+\beta ({x_2}
   {y_2}-{x_1}))(x_2+y_1)}{x_2( {x_2} {y_2}-{x_1}-\beta +1)}, \\  
u_2 &=& -\frac{\alpha ({x_2}+{y_1})^2+(\beta {x_2}+{y_1} {y_2} {x_2}-{x_1} {y_1}+{y_1}) ({x_2} (\beta+{y_1}
   {y_2}-1)-{x_1} {y_1})}{{x_2} ({x_2} {y_2}-{x_1}-\beta+1)}, \\ 
v_1 &=& x_2+y_1, \\ 
 v_2 &=& \frac{{x_1}}{{x_2}}-\frac{\alpha+\beta ({x_2} {y_2}-{x_1})}{{x_2} ({x_2} {y_2}-{x_1}-\beta +1)}.
\end{eqnarray*}
This solution can be obtained by the limit for $\varepsilon \rightarrow 0$ of the $U_{11}$, $U_{12}$, $V_{12}$ 
and $V_{21}$ elements of the matrices 
$$U:=U(L_1'(x_{1},x_{2};\alpha),M'_{\varepsilon}(y_1,y_2,\beta)), \ 
V:=V(L_1'(x_{1},x_{2};\alpha),M'_{\varepsilon}(y_1,y_2,\beta)),$$ 
defined by (\ref{U}) and (\ref{V}) with $L_1', \ M'_{\varepsilon}$ (\ref{L1}) and (\ref{Me}) respectively. 
So the map 
\begin{equation} \label{exmap}
S_{\alpha,\beta}:((x_1,x_2),(y_1,y_2)) \mapsto ((u_1,u_2),(v_1,v_2))
\end{equation}
admits the 
strong Lax pair $(L_1(x_1,x_2;\alpha),L_2(x_1,x_2;\alpha))$. 

Furthermore the twisted equation 
$$L_2(\tilde{u}_1,\tilde{u}_2;\alpha)L_1(\tilde{v}_1,\tilde{v}_2;\beta)=L_1(y_1,y_2;\beta)L_2(x_1,x_2;\alpha)$$ 
implies 
\begin{eqnarray*}
\tilde{u}_1 &=& \frac{(\alpha {x_1}-{y_1} {x_1}-{y_2}) {y_2}}{\beta {x_1}+({y_1}-1) {y_1} {x_1}+(\alpha+{y_1}-1) {y_2}},\\  
\tilde{u}_2 &=& {x_2}-\frac{\beta+({y_1}-1) {y_1}}{{y_2}}, \\ 
\tilde{v}_1 &=& \frac{\beta (\alpha {x_1}-{y_2})+({y_1} {x_1}-{x_1}+{y_2}) (\alpha {y_1}+{x_2} {y_2})+{x_1} {x_2} \left({x_1}
   \left({y_1}^2-{y_1}+\beta\right)+{y_1} {y_2}\right)}{\beta {x_1}+({y_1}-1) {y_1} {x_1}+(\alpha+{y_1}-1) {y_2}}, \\ 
\tilde{v}_2 &=& {x_1}-\frac{(\alpha {x_1}-{y_1} {x_1}-{y_2}) {y_2}}{\beta {x_1}+({y_1}-1) {y_1} {x_1}+(\alpha+{y_1}-1) {y_2}}
\end{eqnarray*}
and the map $T_{\alpha,\beta}$, 
\begin{equation}\label{extw}
{T}_{\alpha,\beta}:((x_1,x_2),(y_1,y_2)) 
\mapsto ((\tilde{u}_1,\tilde{u}_2),(\tilde{v}_1,\tilde{v}_2)), 
\end{equation}
with strong Lax pair $(L_2(x_1,x_2;\alpha),L_1(x_1,x_2;\alpha))$. 

%%%%%%%%%%%%%%%%%%%%%%%%%%%%%%%%%%%%%%%%%%%%%%%%%%%%%%%%%%%%%%%%%%%%%%%%%%%%%%%%%%%%%%%%%%%%%%%%%%%%%%%%%%
%%%%%%%%%%%%%%%%%%%%%%%%%%%%%%%%%%%%%%%%%%%%%%%%%%%%%%%%%%%%%%%%%%%%%%%%%%%%%%%%%%%%%%%%%%%%%%%%%%%%%%%%%%
\subsection{Solutions of the entwining YB equation and integrability}
We consider two cases. First we choose a third matrix $A_3=A_1$ and the corresponding 
Lax matrix $L_3(x_1,x_2;\alpha)=L_1(x_1,x_2;\alpha)$. Theorem \ref{lpun} implies that 
the maps $(S_{\alpha,\beta},R_{\alpha,\beta},T_{\alpha,\beta})$ (\ref{exmap},  \ref{YB1} and \ref{extw}  
respectively) admits the strong Lax triple 
$(L_1(x_1,x_2;\alpha),L_2(x_1,x_2;\alpha),L_1(x_1,x_2;\alpha))$ and satisfy the entwining Yang-Baxter equation  
$$T_{\beta,\gamma}^{23}\circ R_{\alpha,\gamma}^{13}\circ S_{\alpha,\beta}^{12}=
S_{\alpha,\beta}^{12}\circ R_{\alpha,\gamma}^{13} \circ T_{\beta,\gamma}^{23}.$$

Now if choose $A_3=A_2, \ L_3(x_1,x_2;\alpha)=L_2(x_1,x_2;\alpha)$, we derive the strong Lax triple 
$(L_1(x_1,x_2;\alpha),L_2(x_1,x_2;\alpha),L_2(x_1,x_2;\alpha))$ of the maps 
$(S_{\alpha,\beta},S_{\alpha,\beta},\bar{R}_{\alpha,\beta})$, from (\ref{exmap}) and (\ref{YB3}), that satisfy the 
entwining YB equation 
$$\bar{R}_{\beta,\gamma}^{23}\circ S_{\alpha,\gamma}^{13}\circ S_{\alpha,\beta}^{12}=
S_{\alpha,\beta}^{12}\circ S_{\alpha,\gamma}^{13} \circ \bar{R}_{\beta,\gamma}^{23}.$$

The map $S_{\alpha,\beta}$ (\ref{exmap}) is symplectic with respect to the reduced Sklyanin bracket 
on $\Sigma_{A_1}(\alpha,1) \times \Sigma_{A_2}(\beta,1)$:
\begin{equation} \label{pstr}
\{ x_1,x_2 \}=-x_2, \  \{y_1,y_2 \}=1, \ \{x_i,y_j \}=0, \ i,j=1,2 \ \ (x_2 \neq 0).
\end{equation}
For the 1-periodic `staircase' initial value problem  
the trace of the corresponding
monodromy matrix, $M_1(\mathbf{x},\mathbf{y})=L_2(y_1,y_2;\beta)L_1(x_1,x_2;\alpha)$, 
gives the two first integrals:
\begin{eqnarray*}
J^s_1(x_1,x_2,y_1,y_2) &=& x_2 y_2+x_1( y_1 y_2+\beta-1)-
\frac{y_1}{x_2}(x_1^2-x_1+\alpha), \\ 
J^s_2(x_1,x_2,y_1,y_2) &=& x_1+y_1 y_2, 
\end{eqnarray*}
which are in involution with respect to the Poisson bracket (\ref{pstr}). 

Similarly 
$\bar{R}_{\alpha,\beta}$, $T_{\alpha,\beta}$ 
are integrable and $R_{\alpha,\beta}$ is an 
involution i.e. $R_{\alpha,\beta}\circ R_{\alpha,\beta}=Id$. 
The reduced Sklyanin bracket on the corresponding symplectic leaves is 
given by the brackets of the coordinates 
\begin{eqnarray} 
&&\{ x_1,x_2 \}=-x_2, \  \{y_1,y_2 \}=-y_2, \ ~ \{x_i,y_j \}=0, \ \text{on} \ 
 \Sigma_{A_1}(\alpha,1) \times \Sigma_{A_1}(\beta,1), \label{brR}\\ 
&& \{ x_1,x_2 \}=1, \ \ \ \ ~  \{y_1,y_2 \}=1, \ \ ~ ~ \ \  \{x_i,y_j \}=0, \ \text{on} \ 
 \Sigma_{A_2}(\alpha,1) \times \Sigma_{A_2}(\beta,1), \label{brRR}\\
&&\{ x_1,x_2 \}=1, \ \ \ \ ~ \{y_1,y_2 \}=-y_2,  \ ~ \{x_i,y_j \}=0, \ \text{on} \ 
 \Sigma_{A_2}(\alpha,1) \times \Sigma_{A_1}(\beta,1) \label{brT}.  
\end{eqnarray}
The corresponding integrals  
of the YB map $\bar{R}_{\alpha,\beta}$ are    
$$
J^{\bar{R}}_1 = \alpha  y_1 y_2 + b x_1 x_2+(x_2 y_1+1) (x_1 y_2+1),  \ 
J^{\bar{R}}_2 = x_1 x_2+y_1 y_2,  
$$ 
and of the map $T_{\alpha,\beta}$ are   
$$J^T_1 = x_2 y_2+y_1( x_1 x_2 + \alpha -1) - 
\frac{x_1}{y_2}(y_1^2-  y_1+ \beta),  \ 
J^T_2 = y_1+x_1 x_2.$$ 
The integrals of each pair are in involution with respect to the 
Poisson brackets (\ref{brRR}) and (\ref{brT}) respectively.

%%%%%%%%%%%%%%%%%%%%%%%%%%%%%%%%%%%%%%%%%%%%%%%%%%%%%%%%%%%%%%%%%%%%%%%%%%%%%%%%%%%%%%%%%%%%%%%%%%%%%%%
%%%%%%%%%%%%%%%%%%%%%%%%%%%%%%%%%%%%%%%%%%%%%%%%%%%%%%%%%%%%%%%%%%%%%%%%%%%%%%%%%%%%%%%%%%%%%%%%%%%%%%%
%%%%%%%%%%%%%%%%%%%%%%%%%%%%%%%%%%%%%%%%%%%%%%%%%%%%%%%%%%%%%%%%%%%%%%%%%%%%%%%%%%%%%%%%%%%%%%%%%%%%%%%
\section{Perspectives} \label{persp}
We presented the construction of entwining Yang-Baxter maps following the procedure of using Lax matrices 
as symplectic leaves of the Sklyanin bracket of polymonial matrices. These maps give rise to, in principle, asymmetric 
discrete integrable systems on the 2-dimensional integer lattice. We believe an extensive study as well as their relation to scalar asymmetric integrable systems on quad-graphs deserves more attention. 

\section{Acknowledgments}
The first author acknowledges partial support from the State Scholarships Foundation of
Greece.

\section{Appendix A}
The proof of proposition \ref{UV}.

\begin{proof}
From equation (\ref{fact}) we derive the system :
\begin{equation} \label{syst}
UV=YX, \ \ UB+AV=YA+BX
\end{equation}
which implies
\begin{equation} \label{tetr}
(UA^{-1})^2 AB=UA^{-1}(YA+BX)-YX.
\end{equation}
Since $p_{U}^{A}(\zeta)=p_{X}^{A}(\zeta)$, then $f_{i}(U;A)=f_{i}(X;A)$
so from the Cayley-Hamilton theorem we have that:
\begin{equation} \label{cal}
f_{2}(X;A)(UA^{-1})^{2}-f_{1}(X;A)U A^{-1} + f_{0}(X;A)I=0.
\end{equation}
By solving (\ref{tetr}) and(\ref{cal}) with respect to $U$ we get that:
$$U=(f_{2}(X;A)YX-f_{0}(X;A)AB)
(f_{2}(X;A)(YA+BX)-f_{1}(X;A)AB))^{-1}A$$
and from (\ref{syst}) that $V = A^{-1}(YA+BX-U(X,Y)B)$.

From the other hand let us assume that $U$ and $V$ are given by (\ref{U}) and  (\ref{V}). 
By performing some calculations we have that
\begin{eqnarray*}
\Pi _{2}(X,Y) &=& \Pi _{2}(X,Y)+\Pi _{1}(X,Y)A^{-1}X-\Pi _{1}(X,Y)A^{-1}X  \\
&=& \Pi _{1}(X,Y)A^{-1}X - BA(f_{2}(X;A)
(A^{-1}X)^{2}-f_{1}(X;A) A^{-1}X + f_{0}(X;A)I) \\
&=& \Pi _{1}(X,Y)A^{-1}X-BAp_{X}^{A}(A^{-1}X)
\end{eqnarray*}
Cayley-Hamilton theorem states that $p_{X}^{A}(A^{-1}X)=0$ so
$\Pi _{2}(X,Y)=\Pi _{1}(X,Y)A^{-1}X$
and from (\ref{U}) we get $UA^{-1}=\Pi_1A^{-1}X \Pi_1^{-1}$ which means that
$f_{i}(U;A)=f_{i}(X;A)$ or $p_{U}^{A}(\zeta)=p_{X}^{A}(\zeta)$ and
$p_{V}^{B}(\zeta)=p_{Y}^{B}(\zeta)$.
Also from (\ref{U}) and (\ref{V}) we derive that $UB+AV=YA+BX$ and
$$UA^{-1}(f_{2}(X;A)(UB+AV)-f_{1}(X;A)AB)=
f_{2}(X;A)YX-f_{0}(X;A)AB$$
or
$$(f_{2}(X;A)(UA^{-1})^2-f_{1}(X;A)UA^{-1}+f_{0}(X;A)I)AB=
f_2(X;A)(YX-UV).$$
Since $f_i(X;A)=f_i(U;A)$ 
from Cayley-Hamilton theorem we get $UV=YX$.
\end{proof}

%%%%%%%%%%%%%%%%%%%%%%%%%%%%%%%%%%%%%%%%%%%%%%%%%%%%%%%%%%%%%%%%%%%%%%%%%%%%%%%%%%%%%%%%%
%%%%%%%%%%%%%%%%%%%%%%%%%%%%%%%%%%%%%%%%%%%%%%%%%%%%%%%%%%%%%%%%%%%%%%%%%%%%%%%%%%%%%%%%%
%%%%%%%%%%%%%%%%%%%%%%%%%%%%%%%%%%%%%%%%%%%%%%%%%%%%%%%%%%%%%%%%%%%%%%%%%%%%%%%%%%%%%

\section{Appendix B} 
The proof of theorem \ref{lpun}. 
\begin{proof}
The proof of the existence and uniqueness of the solutions of (\ref{uniq}) 
is a consequence of the 
construction of the matrices $L_i'(x_1,x_2,\bar{\alpha})$ in addition 
with Prop. \ref{UV},  
for $X=L_i'(x_1,x_2,\bar{\alpha})$, $Y =L_j'(x_1,x_2,\bar{\alpha})$, 
$A=A_i$ and $B=A_j$.  
The elements $u^{ij}_1, \ u^{ij}_2, \ v^{ij}_1$ and $v^{ij}_2$ 
can be determined by the corresponding elements of 
the matrices 
$$U(L'_i(x_{1},x_{2};\bar{\alpha}),L'_j(y_1,y_2,\bar{\beta})), \ 
V(L'_i(x_{1},x_{2};\bar{\alpha}),L'_j(y_1,y_2,\bar{\beta})),$$ 
of (\ref{U}) and (\ref{V}). Since the map of Prop. \ref{UV}, $(X,Y) \mapsto (U,V)$, is Poisson, 
the reduced map 
$S_{\bar{\alpha},\bar{\beta}}:  \Sigma_{A_1}(\bar{\alpha}) \times \Sigma_{A_2}(\bar{\beta})
\rightarrow  \Sigma_{A_1}(\bar{\alpha}) \times \Sigma_{A_2}(\bar{\beta})$ 
is symplectic (resp. the maps $R_{\bar{\alpha},\bar{\beta}}$ and  
$T_{\bar{\alpha},\bar{\beta}}$). 
 
In order to prove the entwining YB property it suffices to show that the equation 
\begin{eqnarray} \label{pren}
(L_1'(x'_{1},x'_{2};\bar{\alpha})-\zeta A_1) (L'_2(y'_{1},y'_{2};\bar{\beta})-\zeta A_2) (L'_3(z'_{1},z'_{2};\bar{\gamma})-\zeta A_3) 
\nonumber \\ 
=
(L'(x_{1},x_{2};\bar{\alpha})-\zeta A_1) (L'_2(y_{1},y_{2};\bar{\beta})-\zeta A_2) (L'_3(z_{1},z_{2};\bar{\gamma})-\zeta A_3) 
\end{eqnarray} 
implies that $x_i'=x_i, \ y_i'=y_i, \ z_i'=z_i$, for $i=1,2$, then the proof follows 
from Prop. \ref{3product}.
We set 
\begin{eqnarray*}
X=L'_1(x_{1},x_{2};\bar{\alpha}), \ Y=L'_2(y_{1},y_{2};\bar{\beta}), \ Z=L'_3(z_{1},z_{2};\bar{\gamma}), \ \\ 
X'=L'_1(x'_{1},x'_{2};\bar{\alpha}), \ Y'=L'_2(y'_{1},y'_{2};\bar{\beta}), \ Z'=L'_3(z_{1},z_{2};\bar{\gamma}),
\end{eqnarray*}
and  
$XYZ=K$, $XYA_3+XA_2Z+A_1YZ=L$, $XA_2A_3+A_1 YA_3+ A_1 A_2 Z=M$. 
Substituting to (\ref{pren}) we come up to the system :
$$X^{\prime }Y^{\prime }Z^{\prime }=K, \
X^{\prime }Y^{\prime }A_3+X^{\prime }A_2Z^{\prime }+A_1Y^{\prime }Z^{\prime }=L, \
X^{\prime }A_2A_3+A_1 Y^{\prime }A_3+ A_1 A_2 Z^{\prime }=M$$
which implies that 
\begin{equation} \label{tritis}
(X'A_1^{-1})^{3}A_1A_2A_3 -(X'A_1^{-1})^{2}M+X'A_1^{-1}L=K.
\end{equation}
Since $$\det(X'-\zeta A_1)=\det(X-\zeta A_1)= \alpha_2\zeta^{2}- \alpha_1 \zeta + \alpha_0,$$  
with $\alpha_2=f_{2}(X;A_1), \ \alpha_1= f_{1}(X;A_1), \  \alpha_0= f_{0}(X;A_1)$,  
from Cayley-Hamilton theorem we have that
\begin{equation*}
p_{X'}^{A_1}(X'A_1^{-1})=\alpha_2(X'A_1^{-1})^{2}-\alpha_1(X'A_1^{-1})+\alpha_0 I=0.
\end{equation*}
By evaluating the powers of $X'A_1^{-1}$ from the last equation, equation (\ref{tritis}) becomes 
$$X'A_1^{-1}[\alpha_{2}^{2}L-\alpha_{2}\alpha_{1}M+(\alpha_{1}^{2}-
\alpha_{2}\alpha_{0})A_1 A_2 A_3]=
\alpha_{2}^{2}K-\alpha_{2}\alpha_{0}M+\alpha_{1}\alpha_{0}A_1 A_2 A_3$$
or 
\begin{eqnarray} \label{res}
 X'A_1^{-1}[\alpha_{2}^{2}L-\alpha_{2}\alpha_{1}M+
(\alpha_{1}^{2}-\alpha_{2}\alpha_{0})A_1 A_2 A_3] \ \ \ ~~ ~\nonumber \\ = 
XA_1^{-1}[\alpha_{2}^{2}L-\alpha_{2}\alpha_{1}M+(
\alpha_{1}^{2}-\alpha_{2}\alpha_{0}) A_1 A_2 A_3] 
+ Q
\end{eqnarray}
where
\begin{equation*}
Q = \alpha_{2}^{2}K-\alpha_{2}\alpha_{0}M+\alpha_{1}\alpha_{0}A_1 A_2 A_3 -
XA_1^{-1}[\alpha_{2}^{2}L-\alpha_{2}\alpha_{1}M+
(\alpha_{1}^{2}-\alpha_{2}\alpha_{0})A_1 A_2 A_3 ].
\end{equation*}
If we replace again $K, L, M$ by $XYZ$, $XYA_3+XA_2Z+A_1YZ$, 
$XA_2A_3+A_1 YA_3+ A_1 A_2 Z$ respectively, we can factorize $Q$ as follows
\begin{eqnarray*}
Q &=& (\alpha_{2}(XA_1^{-1})^{2}-\alpha_{1}XA_1^{-1}+\alpha_{0}I)
(\alpha_{1}A_1 A_2 A_3-\alpha_{2}(A_1 Y A_3+A_1 A_2 Z)) \\
&=& p_{X}^{A_1}(XA_1^{-1})(\alpha_{1}A_1 A_2 A_3 -\alpha_{2}(A_1 Y A_3+A_1 A_2 Z))
\end{eqnarray*}
and since $p_{X}^{A_1}(XA_1^{-1})=0$, $Q=0$.
So from (\ref{res}) 
we conclude that
$X'=X$. In a similar way we can prove that $Z'=Z$
which means that also $Y'=Y$.
\end{proof}


\begin{thebibliography}{3}


\bibitem{Nich} Brzezi{\'n}ski, T., Nichita, F. \emph{Yang--Baxter systems and entwining structures}, Comm.
Algebra 33, 2005, 1083–-1093.

\bibitem{Frei-Mail}
Freidel, L., Maillet, J.-M., \emph{Quadratic algebras and integrable systems},
Phys.Lett. B, 262:278–284, 1991.


\bibitem{hlav}
Hlavat{\'y} L. \emph{Yang--Baxter systems, solutions and
applications}, arXiv:q-alg/9711027v1, 1997.


\bibitem{hlav2}Hlavat{\'y} L., Snobl L. \emph{Solution of the Yang-Baxter System for Quantum Doubles},
 Int. J. Mod. Phys. A 14 (1999) 3029-3058. 
 
\bibitem{kp1} Kouloukas Th. E., Papageorgiou V. G. \emph{ Yang--Baxter maps with first-degree-polynomial $2 \times 2$ Lax matrices}, 
J. Phys. A: Math. Theor. 42, 2009, 404012 (12pp)

\bibitem{cnp}  Nijhoff, F. W. Capel, H. W., Papageorgiou, V. G. 
\emph{Integrable quantum mappings}, Phys. Lett. A, 46 No 4, 1992, 2155--2158.

\bibitem{pnc}
Papageorgiou, V. G. Nijhoff, F. W. Capel, H. W. \emph{Integrable mappings and
nonlinear integrable lattice equations}, Phys. Lett. A, 147, 1990, 106--114.

\bibitem{ves1}
Reshetikhin N., Veselov A.P. \emph{Poisson Lie groups and
Hamiltonian theory of the Yang--Baxter maps}, 	arXiv:math/0512328v2 , 2005.

\bibitem{skly2}
Sklyanin E.K. \emph{Some algebraic structures connected with the
Yang-Baxter equation}, Funct. Anal. Appl. 16, No 4, 1983, 263--270.

\bibitem{skly3}
Sklyanin, E. K. \emph{B\"acklund transformations and Baxter's
Q-operator},  Integrable systems: from classical to quantum (Montr\'eal, QC, 1999),  227--250, CRM Proc. Lecture Notes, 26, Amer. Math. Soc., Providence, RI, 2000.

\bibitem{ves4}
Suris Yu.B., Veselov A.P.  \emph{Lax matrices for Yang--Baxter maps},  
J. Nonlin. Math. Phys. 10, suppl.2, 2003, 223--230.

\bibitem{ves2}
Veselov A.P. \emph{Yang-Baxter maps and integrable dynamics}, Phys.
Lett. A, 314, 2003, 214--221.

\bibitem{ves3}
Veselov A.P. \emph{Yang-Baxter maps: dynamical point of view},
Combinatorial Aspects of Integrable Systems (Kyoto, 2004), MSJ Mem. vol 17, 2007, 
pp 145--67.

\bibitem{vlad}
Vladimirov A.A. \emph{A method for obtaining quantum doubles from the
Yang--Baxter R-matrices}, 
Mod. Phys. Lett. A, 8, 1993, 1315--1321.

\end{thebibliography}
\end{document}